\documentclass[journal]{IEEEtran}
\ifCLASSINFOpdf
\else
\fi
	\usepackage{graphicx}
	\usepackage{epstopdf}
	
	\usepackage{amsthm}
	\usepackage{amsmath}
	\usepackage{amsfonts}
	\usepackage{cases}
	\usepackage{amssymb}
	\usepackage{yhmath}
	\usepackage{cite}
	\usepackage{epstopdf}
	\usepackage{booktabs}
	\usepackage[ruled]{algorithm2e}
	\usepackage{subfigure} 
	\usepackage{bm}
	\usepackage{svg}
\begin{document}
		\theoremstyle{definition}
		
		\newtheorem{proposition}{\indent Proposition}
		\newtheorem{remark}{\indent Remark}
		\newtheorem{lemma}{\indent Lemma}
		
	\title{\Large{Holographic Airy Beamforming: Curved Trajectory Optimization for Blockage-Resilient Terahertz Communications}\\
		{
		}

		\author{\IEEEauthorblockN{Xinyuan~Hu$^{\ast}$,
				Boya~Di$^{\ast}$,	and Lingyang~Song$^{\ast}$}\\
			
			\IEEEauthorblockA{$^{\ast}$School of Electronics, Peking University, Beijing, China.\\
				}
			
			Email: $^{\ast}$\{huxiny, diboya, lingyang.song\}@pku.edu.cn
			
		}
	}
	
		

		
		
		\maketitle
		
		\begin{abstract}
			Terahertz communication offers vast bandwidth for high-speed transmission in the 6G networks but faces severe blockage challenges in the near-field region due to large antenna arrays. To overcome the limitation that near-field focused beams are susceptible to obstacles, wavefront engineering is leveraged to generate an Airy beam that propagates along a parabolic trajectory to circumvent blockages. In this paper, we consider the reconfigurable holographic surface (RHS) as a potential  solution for such precise  wavefront engineering  owing to its compact radiation element spacing being much smaller than half-wavelength. We reveal that the adjustable effective aperture of the RHS allows the parabolic offset to be located within the antenna aperture, which enhances the freedom in designing Airy beam trajectories. An analog beamforming method, named the holographic Airy beamforming scheme based on amplitude control, is then proposed to generate the curved beam that propagates along the desired trajectory. To maximize the received power of a blocked user, we develop a geometry-based trajectory optimization algorithm. Simulation results validate that, compared to traditional phase-controlled arrays with analog beamforming, the RHS can leverage its adjustable effective aperture to improve the received power of the blocked user by over 10 dB.

		\end{abstract}
		
		\begin{IEEEkeywords}
			Reconfigurable holographic surface, Airy beams, Curved beams
		\end{IEEEkeywords}


		%
		\IEEEpeerreviewmaketitle

		\section{Introduction}
	
	By exploiting the vast bandwidth in the 0.1–10 THz spectrum, terahertz (THz) communication has become a key technology for next-generation wireless systems, enabling high-speed transmission~\cite{ref.Thz_commun}. To compensate for the severe path loss, it necessitates the use of extremely large-scale antenna arrays (ELAAs) to form directional beams. The ELAAs, combined with the use of high-frequency bands, results in the extension of the near-field region, thereby placing more users and obstacles within it. In this context, traditional near-field beam focusing, which concentrates energy into a narrow spot, is highly susceptible to blockages due to the poor penetration and strong absorption of THz signals~\cite{ref.near_field}. Consequently, THz communication face  challenges in ensuring high signal quality with low overhead in blockage environments.
	
	To address this challenge, wavefront engineering, which refers to the precise shaping of the electromagnetic (EM) field at the transmitter, emerges as a new paradigm for navigating complex environments. This approach enables the generation of non-diffracting (ND) beams with inherent self-healing properties, allowing wavefront reconstruction after encountering obstacles to enhance link robustness. However, practical constraints such as limited antenna aperture and finite energy make ideal ND beams difficult to realize. As a result, for ND beams propagating along straight trajectories, like Bessel beams, both self-healing capability and obstacle resistance become limited when blocked~\cite{ref.nature_bessel}. In comparison, Airy beams, as another type of ND beams, exhibit self-acceleration, enabling them to propagate along pre-designed \emph{parabolic trajectories} without external potentials. This capability allows them to detour around obstacles in the line-of-sight (LoS) path, establishing reliable communication links and presenting a promising strategy for robust wireless connectivity in obstructed environments without auxiliary relays~\cite{ref.airybeam_basic}.
	
	Recent research on the bending curves of Airy beams has been actively conducted in various fields of wireless communications. The work~\cite{ref.airy_learning_NC} has employed a pre-trained neural network to generate optimized near-field Airy beams capable of supporting GHz-level bandwidth. In~\cite{ref.airy_codebook} An attention-based multi-parameter beam training network has been proposed to jointly predict the angle, distance, and curvature parameters of the optimal Airy beam in the airy codeword. The work~\cite{ref.airy_qinyf}  has demonstrated that geometry-based Airy beamforming can effectively deliver power into shadowed areas and restore multi-user channel rank, offering a solution for blockage mitigation in near-field communications. 
	


	However, existing research generates Airy beams using traditional phased arrays, whose element spacing is typically half a wavelength, thereby constraining the precision of wavefront engineering. In contrast, the reconfigurable holographic surface (RHS) is an energy-efficient and ultra-thin metamaterial transceiver capable of achieving sub-wavelength element spacing, which enables higher-resolution wavefront engineering~\cite{ref.zsp_movable_antenna}. Moreover, we reveal that the RHS features an adjustable effective aperture. By selectively activating elements via amplitude control, the beam’s initial offset can be shifted within the physical aperture. This provides an additional degree of freedom to optimize the curved trajectory relative to the obstacle's geometry.
	
	
	In this paper, we  investigate  an RHS-enable downlink communication system where the RHS-equipped base station transmit signal to a user blocked by an obstacle. To avoid signal blockage, we aim to generate a curved beam that can bypass the obstacle. This task is challenging because Airy profiles are conventionally generated through cubic phase modulation, whereas RHS operates on a holographic amplitude-control principle.
	To address this challenge, we contribute to RHS-enabled wireless communications as follows. \emph{First}, we propose a holographic Airy beamforming scheme based on amplitude modulation, which generates a curved beam that propagates along a predefined trajectory. \emph{Second}, leveraging the RHS’s adjustable effective aperture, we demonstrate that optimizing the beam's starting position and curvature jointly provides superior bypass capabilities compared to fixed-aperture phase-controlled arrays. Simulation results confirm that this tunability allows the RHS to provide the user with enhanced communication performance in complex blockage scenarios.
	
	
		\section{System Model}

	In this section, we first introduce the transmission scenario and the working principle of the RHS. Following this, the principles of holographic Airy beamforming is elaborated and the EM wave propagation model is presented.
		\subsection{Scenario Description and Working Principle of RHS}
		
		As shown in Fig.~\ref{fig.system_model}, we consider an RHS-enabled wireless communication system, where a base station (BS) equipped with an RHS transmits signals to a single-antenna user, with an obstacle present in the LoS link between the array and the user.
		Specifically, within a two-dimensional Cartesian coordinate system $(x, z)$,  the user is located at $(x_r,z_r)$. The obstacle, which completely blocks the direct  LoS path from the user to the RHS, is modeled as a rectangle with vertex $(x_{obs},z_{obs})$, length $l_{z,obs}$, and width $l_{x,obs}$.
		
		We consider a one-dimensional RHS with an aperture of $[0,l_{RHS}]$ positioned along the $x$-axis. The RHS comprises a feed, a matched load, and $N$ radiation elements embedded in waveguides with a uniform spacing of $d$. In operation, the feed injects an RF signal that has a power of $P_t$, serves as the reference wave, propagates along the waveguide, and sequentially excites each radiation element. At element $n$, tunable components regulate the proportion of energy coupled from the reference wave, i.e., the radiation ratio $\eta_n$ (ideally continuously tunable between 0 and 1), thereby controlling the radiated amplitude. The radiation phase is inherently determined by the propagation phase delay of the reference wave from the feed to the $n$-th element. The residual energy not coupled at the $n$-th element propagates forward to excite the ($n+1$)-th element and is ultimately absorbed by the matched load. Based on the working principle described above, the EM wave radiated by the $n$-th element can be given by
		\begin{equation}
			\label{eq.rad_element}
			M_{n} = \sqrt{\prod\nolimits_{k=1}^{n-1}(1 - \eta_{k})\cdot \eta_{n}}\cdot \sqrt{P_t}\cdot e^{-j\mathbf{k}_s\cdot\mathbf{r}_{n}^{s}},
		\end{equation}
		where $\mathbf{r}_{n}^{s}$ is the position vector from the feed to the $n$-th element, $\mathbf{k}_s$	is the propagation vector of the reference wave. Therefore, by shaping the radiation ratio distribution via 
		$\{\eta_n\}$, the RHS achieves precise wavefront engineering for beamforming.

	Under the theoretical condition that the radiation ratio $\eta_{n}$	can be continuously tunable within $[0, 1]$, the radiation amplitude of the $n$-th element, given by $\sqrt{\prod\nolimits_{k=1}^{n-1}(1 - \eta_{k})\cdot \eta_{n}}$, is equavlent to $\sqrt{\eta_{eq}}\cdot m_n\cdot s_n$ (see Appendix~A for proof). Here, $s_{n} \in\{0,1\}$ is the radiating activation state of element $n$, $m_{n}\in[0,1]$ is the tunable amplitude for active elements. The equivalent radiation ratio $\eta_{eq}$ is defined as the ratio of the power received by the element to the total power of the reference wave, and it satisfies $\sum\eta_{eq}m_n^2s_n \le 1$. Thus,  the EM wave radiated by the $n$-th element can be rewritten as
		\begin{equation}\label{eq.RHS_model}
		M_{n} = \sqrt{\eta_{eq}}\cdot m_n\cdot s_n\cdot \sqrt{P_t}\cdot e^{-j\mathbf{k}_s\cdot \mathbf{r}_{n}^{s}}.
	\end{equation}

	\subsection{Principles of Holographic Airy Beamforming}
	\begin{figure}[t]
     \setlength{\abovecaptionskip}{-4pt}
		\centering
		\includegraphics[width=0.7\columnwidth]{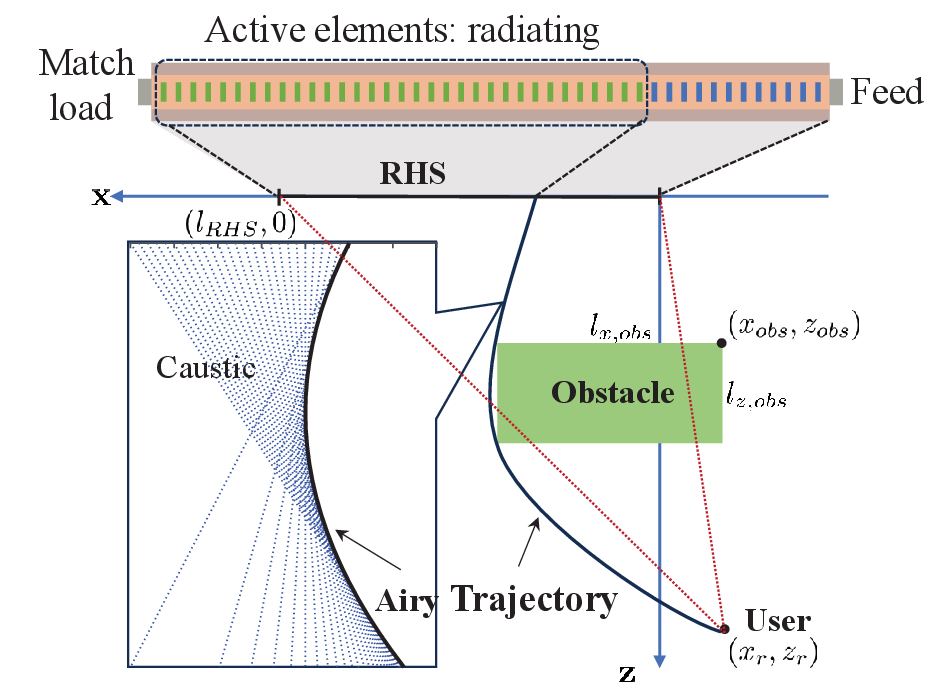}
		\caption{RHS-enabled downlink communication system in an obstructed environment}
		\label{fig.system_model}
		\vspace*{-10pt}
	\end{figure}

	To circumvent obstacles in the LoS link between the RHS and the user, we consider  generating a curved beam that propagates along a parabolic trajectory, drawing on the bending property of Airy beams. A method to generate such a beam, bending along a tunable parabolic trajectory   $x = f(z) = az^2 + bz + c$, is introduced in~\cite{ref.bending_beam}. This method treats  the antenna aperture as a continuous array comprising an infinite number of infinitesimal sub-arrays. Each sub-array emits a ray, and the envelope of all such rays forms a caustic, which defines the desired parabolic trajectory. 
	To synthesize the caustic corresponding to the trajectory $x=f(z)$, a phase profile $\phi(x)$ is designed across the aperture, which should satisfy the condition
	\begin{equation}
			\vspace*{-2pt}
		\label{eq.function_need_to solve}
		\frac{d \phi(x)}{dx}  = k_f\frac{df(z)/dz}{\sqrt{1 + (df(z)/dz)^2}},
	\end{equation}
	where $k_f$ is the free-space wave number.
	
	\begin{proposition}
		\label{pro.phase_profile}
		For a given parabolic trajectory $x = f(z) = az^2 + bz + c (a\neq 0)$, the phase profile $\phi(x)$ across the aperture that satisfies (\ref{eq.function_need_to solve}) is given by	
		\begin{equation}
			\label{eq.phase_profile}
			\vspace*{-2pt}
		 \frac{k_f}{4a}\left[\operatorname{arcsinh}\left(\psi(x)\right)+
			(2b - \psi(x))\sqrt{\psi^2(x) + 1}
			\right],
		\end{equation}
		where $\psi(x) =2a\sqrt{\frac{c-x}{a}} + b$, valid when $a(c-x)\ge 0$.
	\end{proposition}
	\begin{IEEEproof}
		See Appendix B.
	\end{IEEEproof}

	 Based on the phase profile derived from Proposition~\ref{pro.phase_profile} that generates the parabolic trajectory, we implement holographic Airy beamforming by applying this profile to the RHS amplitude control model in (\ref{eq.RHS_model}) using the conventional holographic beamforming method. 
	Specifically, we define the phase difference at the $n$-th RHS element between the phase profile of parabolic beam $\phi((n-1)d)$ given by (\ref{eq.phase_profile}) and the reference wave phase shift $-\mathbf{k}_s\cdot\mathbf{r}_{n}^s$ given by  (\ref{eq.RHS_model}) as 	$\Delta\Phi_n =\phi((n-1)d) + \mathbf{k}_s\cdot\mathbf{r}_{n}^s $.
	 The radiation power of each element is designed to be inversely proportional to this phase difference, ensuring that a small amount of energy is radiated when $\Delta\Phi_n$ is large and a large amount is radiated when it is small. The function $\frac{\cos{(\Delta\Phi_n)} + 1}{2}$ naturally fulfills this amplitude control requirement while also respects the tuning range of the element's amplitude  $m_n$. Therefore, the target bending beam is synthesized via this amplitude modulation. The corresponding radiated field at the 
	 $n$-th element is
	 \begin{equation}
	 	\label{eq.rhs_radiated_field}
	M_{n} = \sqrt{\eta_{eq}}\cdot\frac{\cos{(\Delta\Phi_n)}+1}{2}\cdot s_{n}\cdot\sqrt{P_t}\cdot e^{-j\mathbf{k}_s\cdot \mathbf{r}_{n}},
\end{equation}
where $s_n = 1$ when $a[c-(n-1)d]\ge0$,and $s_n = 0$ otherwise.

		


		\subsection{EM Wave Propagation Model under Blockage}
				
		Based on the initial radiated electric field given by (\ref{eq.rhs_radiated_field}) at the $z=0$ plane, we employ an iterative angular spectrum method (ASM) rooted in Rayleigh-Sommerfeld (RS) theory to compute the field at the user's location~\cite{ref.airy_codebook}.  According to RS theory, the field at the plane $z$ results from the superposition of spherical waves emanating from every point on the wavefront at the preceding plane $z' = z - \delta_z$, with each point acting as a secondary source. Denoting $E_z(x)$ as the electric field at the plane $z$, we have
		\begin{equation}
			\label{eq.RS_propagates}
			E_z(x) = \int E_{z'}(x')\left[ \frac{\delta_z e^{-jk_fr}}{2\pi r^2}(jk_f + \frac{1}{r})\right]dx',
		\end{equation}
		where $r = \sqrt{(x-x')^2 + \delta_z^2}$ is the distance between two points on the two successive propagation planes.
		
	Benefiting from the convolution form in (\ref{eq.RS_propagates}), the ASM applies a Fourier transform to convert the integral into a multiplication to simplify the computation of the field across propagation planes. The Fourier transform of the field is given by $\mathcal{F}\{	E_z(x)\} = \mathcal{F}\{E_{z'}(x)\}	\times H_{\delta_z}$, where $H_{\delta_z} = \mathcal{F}\{\frac{1}{2\pi}\frac{\partial}{\partial z'}(e^{-jk_fr}/r)\}$ is the transfer function dependent on the propagation geometry and wavenumber. By applying the inverse Fourier transform, the field at the $z$-plane is then given by $E_z(x) = \mathcal{F}^{-1}\{\mathcal{F}\{E_{z-\delta_z}(x)\}	\times H_{\delta_z}\}$.	To incorporate the effect of obstacles, we introduce the blockage parameter $B_z(x)$. For a point $(x,z)$,  $B_z(x) = 1$  indicates free space, while $B_z(x) = \alpha \in [0,1)$ represents the attenuation factor when the point is blocked. Thus, the electric field satisfies
		\begin{equation}
			\label{eq:ASM_blockage}
			E_z(x) = B_z(x)\mathcal{F}^{-1}\{\mathcal{F}\{E_{z-\delta_z}(x)\}	\times H_{\delta_z}\}  .
		\end{equation}
		
		We define the operator $(\cdot,\cdot)\otimes$ via  $(B_z(x),H_{\delta_z})\otimes E_{z-\delta_z}(x) = B_z(x)\mathcal{F}^{-1}\{\mathcal{F}\{E_{z-\delta_z}(x)\}	\times H_{\delta_z}\}$, describing a single ASM propagation step given in (\ref{eq:ASM_blockage}). By dividing the path from the RHS at the $z = 0$ plane to the user at the $z = z_r$ plane  into $S$ propagation planes with spacing $\delta_z$,  the electric field at the $z = z_r$ plane can be calculated iteratively as
		\begin{equation}
			\label{eq.E_rev}
				E_{z_r}(x) = (B_{z_r}(x), H_{\delta_z}) \otimes \cdots  \otimes(B_{\delta_z}(x), H_{\delta_z}) \otimes E_0(x).
		\end{equation}

Accordingly, the electric field $E_r$ at the user is  given by 
		\begin{equation} 
			E_r = E_{z_r}(x)|_{x = x_r} = \mathbf{H}\mathbf{M}|_{x = x_r} ,
		\end{equation}
		where  $\mathbf{H} = \prod_{s = 1}^{S}(B_{s\delta_z}(x),H_{\delta_z}) \otimes$ is the  equivalent channel encompassing the blockages, and the beamforming vector $\mathbf{M}$ is the initial electric field $E_0(x)$ in which  $E_0((n-1)d)$ is the radiated field of the $n$-th  RHS given by $M_n$ in (\ref{eq.rhs_radiated_field}). It follows that the received power at the user can be obtained by
		\begin{equation}\label{eq.rev_power}
			P_r =\frac{A_e\cdot|E_r|^2}{Z_0} = \frac{A_e\cdot |\mathbf{H}\mathbf{M}|_{x = x_r}^2}{Z_0},
		\end{equation}
		where $Z_0$ is the impedance of free space and $A_e$ is the effective aperture of the receiving antenna.

	
		\section{Key Idea of Adjustable effective Aperture in Holographic Airy Beamforming}
		
		
		%
		 In this section, we compare the constraints on curve parameters when using RHS versus phase-controlled arrays as transmitting antennas and demonstrate how RHS allows for more flexible curve parameter design.  The examples are subsequently provided to show the advantage of the RHS's adjustable effective aperture in generating curved beams.
		
			
		\subsection{Effective Aperture Adjustment for Trajectory Parameter Flexibility}
			\begin{figure}[t]
			\setlength{\abovecaptionskip}{-4pt}
			\centering
			\includegraphics[width=0.7\columnwidth]{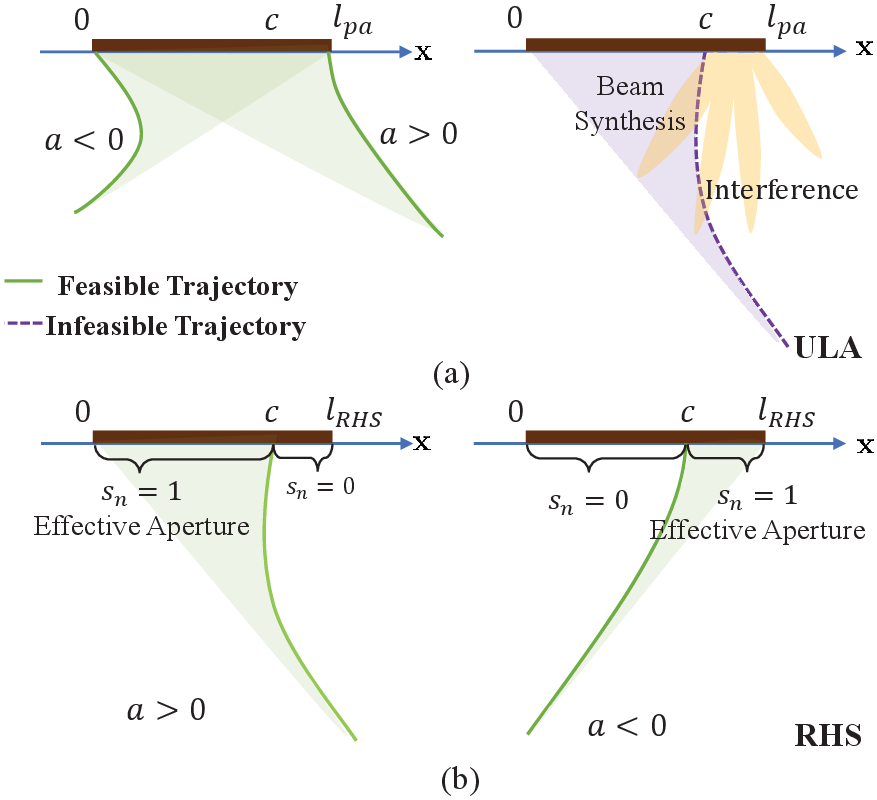}
			\caption{Illustration of trajectory design constraints for: (a) ULA; (b) RHS.}
			\label{fig.ill_aperture}
			\vspace{-10pt}
		\end{figure}

		Given a fully blocked LoS path as described in the system model, we aim to generate a curved beam that follows a parabolic trajectory (i.e., $x = az^2 + bz + c$) to circumvent the obstacle. In the case of a phase-controlled uniform linear array (ULA) with a fixed aperture of $[0, l_{pa}]$  employed as the transmitter,  the beamforming vector can be expressed as $\frac{\sqrt{P_t}}{\sqrt{N_{pa}}}[e^{j\phi(0)},\cdots,e^{\phi(l_{pa})}]^{T}$,which is derived from the phase profile $\phi(x)$ given in (\ref{eq.phase_profile}).	According to Proposition~\ref{pro.phase_profile}, a valid phase solution at an aperture point $x$ exists only if $a(c-x)\ge0$. Consequently, the entire ULA aperture can follow this profile to generate the curved beam only when the offset 
	$c$ satisfies $c>l_{pa}$ for $a>0$, or $c<0$ for $a<0$.  If the desired trajectory has an offset $c$ lying inside the aperture, as illustrated in Fig.~\ref{fig.ill_aperture}(a) for the case of $a>0$, only the aperture of $[0,c]$ can generate the desired curved beam using the phase profile from (\ref{eq.phase_profile}). The remaining aperture $[c,l_{pa}]$ would, however, radiate field that causes interference and distorts the beam synthesis. Therefore, for a ULA, the offset $c$ cannot lie within the aperture interval~$[0,l_{pa}]$.

		
		
		 
		 
		When the RHS with an aperture of $[0, l_{RHS}]$ serves as the transmitter, its amplitude-modulation capability allows individual elements to be  deactivated, i.e., the radiating activation state $s_n$ modeled by (\ref{eq.RHS_model}) is set to 0.  Accordingly, the effective aperture constituted by radiating elements (i.e., $s_n =1$)  can be dynamically adjusted within $[0, l_{RHS}]$. This flexibility allows the generation of a curved beam even when the trajectory offset $c$ lies inside the aperture interval.  Specifically, as shown in Fig.~\ref{fig.ill_aperture}(b), the effective radiating aperture is adjusted to $[0,c]$	for  curvature $a>0$, and $[c,l_{RHS}]$ for curvature $a<0$.
		 
		 Moreover, the RHS elements do not perform amplitude modulation within a fixed range but rather adjust the radiation ratio. Therefore, varying the effective aperture does not change the total radiated power. When the aperture is reduced, each activated element operates with a higher radiation ratio, while when the aperture is enlarged, the radiation ratio per element decreases. By utilizing amplitude control via radiation ratio,  the RHS can dynamically adjust its effective aperture to support holographic Airy beamforming while maintaining a constant total radiated energy to ensure consistent communication performance.

		 
		
		\begin{figure}[t]
			\centering
			\includegraphics[width=0.9\columnwidth]{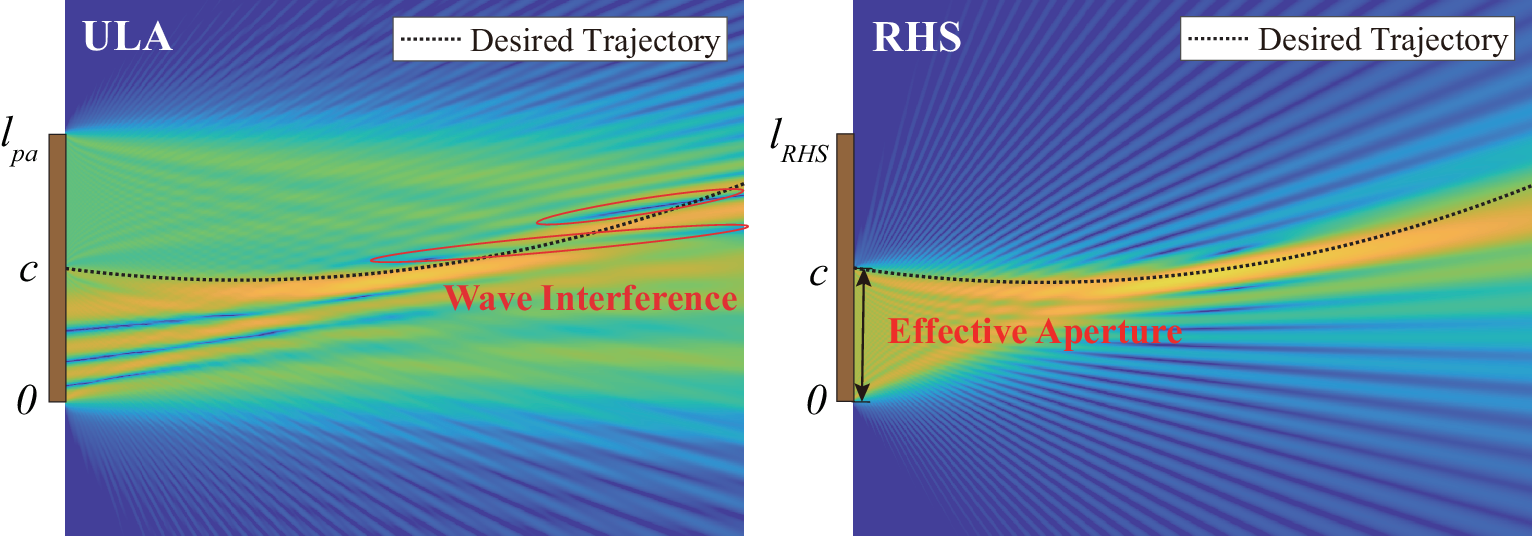}
			\caption{ Comparison of normalized power distribution between RHS and ULA for trajectory offsets $c$ within the aperture }
			\label{fig.aperture2}
			\vspace*{-10pt}
		\end{figure}
		\begin{figure}[t]
			\setlength{\abovecaptionskip}{-4pt}
			\centering
			\includegraphics[width=\columnwidth]{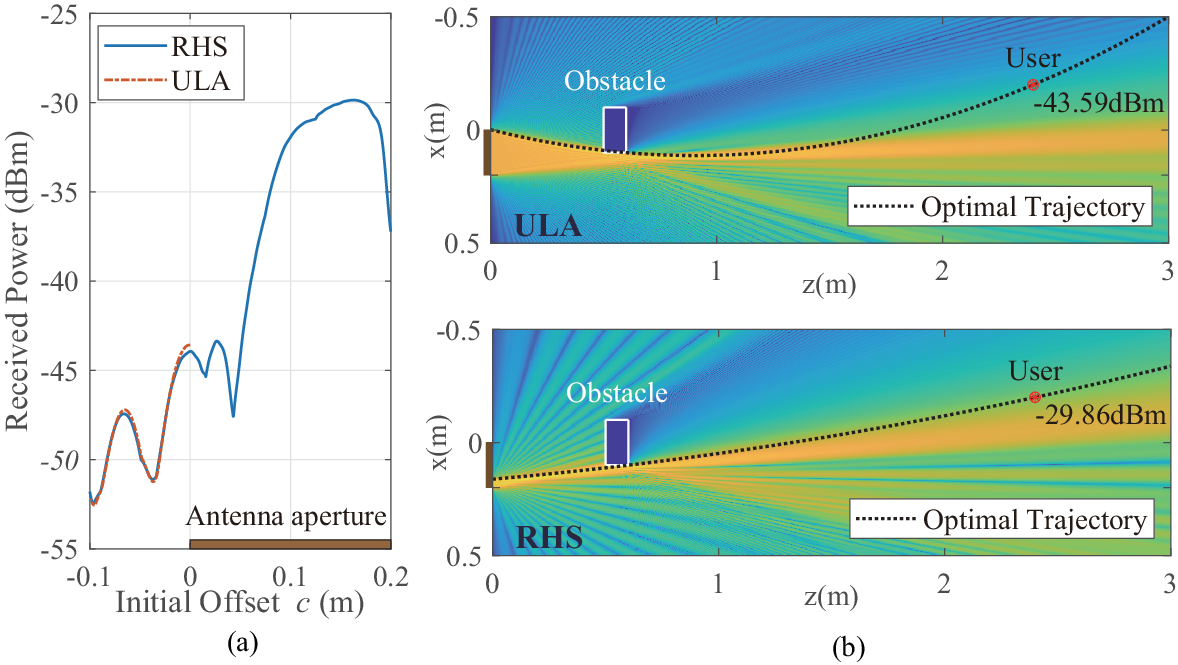}
			\caption{Comparison between RHS and ULA: (a) Power received at the user vs. the initial offset $c$ of the trajectory  (b) Power distribution of the optimal curved beam}
			\label{fig.aperture}
			\vspace*{-10pt}
		\end{figure}


		
		
		\subsection{Illustrative Examples}
		
		Fig.~\ref{fig.aperture2} illustrates the power distribution of the radiated field for both the ULA and the RHS when the offset $c$ of the desired trajectory lies within the aperture.  For the ULA, while the aperture $[0,c]$ can in principle generate the desired curved trajectory, its radiation destructively interferes with that from $[c,l_{pa}]$, preventing the formation of a desired curved beam. In contrast, the RHS adjusts its effective aperture to $[0,c]$, which enables the synthesis of a curved beam propagating along the desired trajectory. Consequently, the RHS allows the optimization of the trajectory offset $c$ within its aperture.
		
		Fig.~\ref{fig.aperture} compares the performance of the RHS and the ULA in generating curved beams to circumvent an obstacle.	The line-of-sight (LoS) path between the user located at $(-0.2, 2.4)$~m and the antenna aperture $[0, 0.2]$~m is blocked by an obstacle spanning $[-0.1, 0.1]$~m $\times [0.5, 0.6]$~m. As shown in Fig.\ref{fig.aperture}(a), when the offset $c$ of the parabolic trajectory lies outside the aperture,  both the RHS and the ULA can generate a curved beam that delivers comparable received energy to the user. However, when the offset $c$ lies within the aperture interval, only the RHS can generate a valid curved beam due to its capability for dynamic effective aperture adjustment. This provides greater flexibility in beam trajectory design and enhances the user's received power. Fig.\ref{fig.aperture}(b) shows the power distribution along the trajectory that maximizes the received power, as obtained from Fig.\ref{fig.aperture}(a).	It demonstrates that by supporting beam trajectories with $c$ inside the aperture, the RHS achieves a smaller curvature than the ULA, thus enabling farther energy propagation and a gain of over 10 dB in received power.



		\section{Parabolic Trajectory Design in the Fully Blocked Scenario}
		
		In this section, we formulate the trajectory optimization problem for maximizing the received power at the user. An efficient geometry-based algorithm is then proposed, which reduces the search dimension through geometric constraints and leverages analytical insights to estimate and refine the optimal parameter.
		\subsection{Problem Formulation}
		In the blocked RHS-aided communication system, we aim optimize the trajectory  $x = f(z) = az^2 + bz+c$ of the curved beam to maximize the user's received power, which can be formulated as
		\begin{align}
			\label{problem}
			\max_{a,b,c} \quad & P_r =\frac{A_e}{Z_0}|E_r|^2 = \frac{A_e}{Z_0}|\mathbf{H}\mathbf{M}|_{x = x_r}^2 \\
			\text{s.t.} \quad &
			\begin{cases} 
				c \le l_{RHS}, & a>0 \\
				c \ge 0, & a<0
			\end{cases}\tag{\ref{problem}a}.
		\end{align}
	
	However, the optimization of the trajectory parameters $a$, $b$, and $c$ involves a substantial search space within a three-dimensional domain. To narrow the feasible region, we impose a constraint that the curved trajectory is supposed to pass through both the user's location  $(x_r, z_r)$ and the obstacle point $(x_o, z_o)$ that requires circumvention. This constraint effectively reduces the dimension of the optimization problem by one, such that only the offset $c$ requires direct optimization. For any given value of $c$, the corresponding  values for $ a$ and $b$ can be derived~as
\begin{equation}\label{eq.ab_c}
	\begin{bmatrix} a \\ b \end{bmatrix}
	=
	\begin{bmatrix}
		z_r^2 & z_r \\
		z_o^2 & z_o
	\end{bmatrix}^{-1}
	\begin{bmatrix} x_r - c \\ x_o - c \end{bmatrix}.
\end{equation}

Therefore, the received power maximization problem (\ref{problem}) can be simplified as
\begin{equation}
	\label{prob.simple}
	\max_{c} \quad   P_r \quad s.t. \quad (\ref{problem}a),(\ref{eq.ab_c})
\end{equation}

		\subsection{Trajectory Optimization}
	
		Since the equivalent channel $\mathbf{H}$ in the obstructed environment  needs to be calculated through a layered iterative procedure, a direct numerical search for the optimal $c$ in (\ref{prob.simple}) is computationally prohibitive. To enable efficient optimization,  we next estimate the optimal value of $c$, denoted as $c_{est}$, by leveraging geometric insights of the trajectory. Serving as a reference,  $c_{est}$ confines the subsequent search to a promising region for  determining the optimal $c_{opt}$.

	 As illustrated in Fig.~\ref{fig.algorithm}, the curved beam is formed by a caustic consisting of the tangent rays of the trajectory. The maximum propagation distance $z_{max}$  over which the beam energy remains concentrated along the trajectory is determined by the farthest tangent ray that can be generated by the antenna aperture, yielding
		\begin{equation}\label{eq.zmax}
			z_{max} =
			\begin{cases}
				\sqrt{c/a}, & a>0 \\
				\sqrt{(c - l_{RHS})/a}, & a<0
			\end{cases}.
		\end{equation}

		\begin{figure}[t]
			\centering
			\includegraphics[width=0.7\columnwidth]{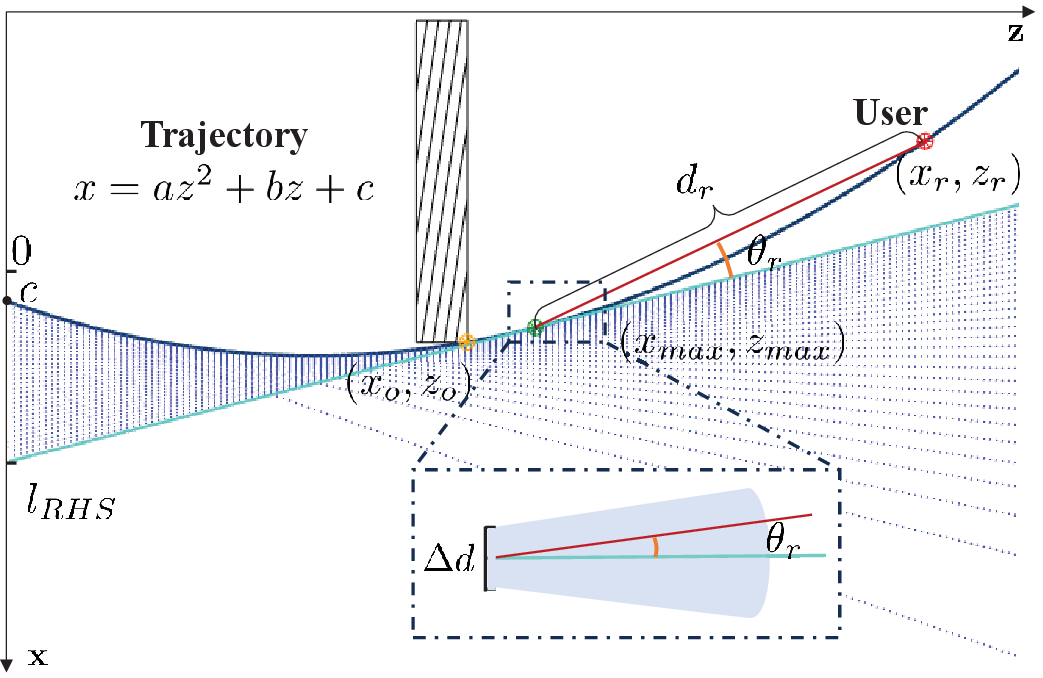}
			\caption{Schematic of the geometric structure of a parabolic trajectory beam in obstructed environment }
			\label{fig.algorithm}
		\vspace{-10pt}
		\end{figure}
		
		\begin{remark}\label{remark~1}
			
		
		In the fully blocked scenario, the obstruction of the LoS path prevents the aperture from generating a tangent ray of the obstacle-detouring curved trajectory that passes through the user's location $(x_r,z_r)$. Hence, the maximum propagation distance  $z_{max}$	is always less than $z_r$, placing the user within the energy-decay region where diffraction dominates.
		\end{remark}

      
      To estimate the user's received power in the energy-decay region, we consider the user is located at the far-field region of the diffraction Gaussian beam generated at  $(x_{max} = f(z_{max}),z_{max})$ with a waist size $\Delta d$. According to the Friis formula, the approximated received power at the user is proportional to 
      \begin{equation}
      	\label{eq.approx}
      	(\frac{\lambda}{4\pi d_r})^2\cdot e^{-\frac{2\pi^2\Delta d^2}{\lambda^2}\theta_r^2},
      	\end{equation}
      where $d_r = \sqrt{(x_r - x_{max})^2 + (z_r - z_{max})} $ is the distance between the user and the point $(x_{max},z_{max})$, $\lambda$ is the wavelength, $e^{-\frac{2\pi^2\Delta d^2}{\lambda^2}\theta_r^2}$ is the radiation pattern of the Gaussian beam~\cite{ref.gaosibeam}, and $\theta_r$ is the angle of the user relative to the tangent to the trajectory at $z = z_{max}$, given~by
      \begin{equation}
      	\label{eq.theta_r}
      	\theta_r = \arctan{\left(\left| \frac{x_r - x_{max} - (z_r - z_{max})(2az_{max} + b)}{z_r - z_{max} +(2az_{max} + b) (x_r - x_{max})} \right|\right)}.
      \end{equation}
       
       From (\ref{eq.ab_c}), (\ref{eq.zmax}), and (\ref{eq.theta_r}), it can be seen that both $d_r$ and $r_{\theta}$ can be expressed as functions of $c$. Consequently, the expression (\ref{eq.approx}), which approximates the power at the user's location, is an  explicit function of $c$, denoted as $\Lambda(c)$. Let $c_{est} = \arg \max_{c} \Lambda(c)$ serve as the starting point of searching for the optimal $c_{opt}$. The search then identifies $c_{opt}$ as the nearest point to $c_{est}$ where a local maximum is achieved. The detailed algorithm is summarized in Algorithm~1. With the optimized offset $c_{opt}$ obtained by Algorithm~1, the holographic Airy beamformer $\mathbf{M}$ can be derived from (\ref{eq.rhs_radiated_field}) and ({\ref{eq.ab_c}).
       
    
    \begin{algorithm}[t]
    	\caption{Geometry-based Trajectory Optimization}
    	\label{algorithm1}
    	\KwIn{User's location $(x_r,z_r)$; obstacle point $(x_o,z_o)$}
    	\textbf{Initialize:}Solve $c_{est} = \arg \max_{c} \Lambda(c)$; Compute $P_{est}$ by (\ref{eq.ab_c}), (\ref{eq.rhs_radiated_field}) and (\ref{eq.rev_power}) \\
    	 Set $P_{1}^{*} = P_{1} = P_{est}$, $c_{1}^{*} = c_{1} = c_{est}$\\
    	\While{$P_{1}\ge P_{1}^{*}$}{
    		$P_{1}^{*}\leftarrow  P_1$,	$c_{1}^{*}\leftarrow  c_1$,	$c_1  \leftarrow c_1 - \delta_c$,\\
    		Compute $a_1,b_1$ by (\ref{eq.ab_c})\\
    		\uIf{constraint (\ref{problem}a) is not satisfied}{
    			break
    		}
    		\Else{
    			Compute $P_1$ by (\ref{eq.rhs_radiated_field}) and (\ref{eq.rev_power})
    		}
    	}
    		
    	Set $P_{2}^{*} = P_{2} = P_{est}$, $c_{2}^{*} = c_{2} = c_{est}$\\
    	\While{$P_{2}\ge P_{2}^{*}$}{
    			$P_{2}^{*}\leftarrow I_2$,	$c_{2}^{*}\leftarrow c_2$, $c_2  \leftarrow c_2 + \delta_c$,\\
    		 Compute $a_2,b_2$ by (\ref{eq.ab_c}) \\
    		\uIf{constraint (\ref{problem}a)  is not satisfied}{
    			break
    		}
    		\Else{
    			Compute $P_2$ by (\ref{eq.rhs_radiated_field}) and (\ref{eq.rev_power})
    		}
    	}
    $P_{opt} = \max{ \{P_1^{*},P_2^{*}\}}$
 
    	\KwOut{The optimal  recieved power $P_{opt}$ }
    \end{algorithm}
    
    \vspace*{-2pt}
		
		\section{Simulation Results}

			\begin{figure*}[ht]
			\begin{minipage}[b]{.32\linewidth}
				\setlength{\belowcaptionskip}{-15mm}
				\centering
				\includegraphics[width=0.99\linewidth]{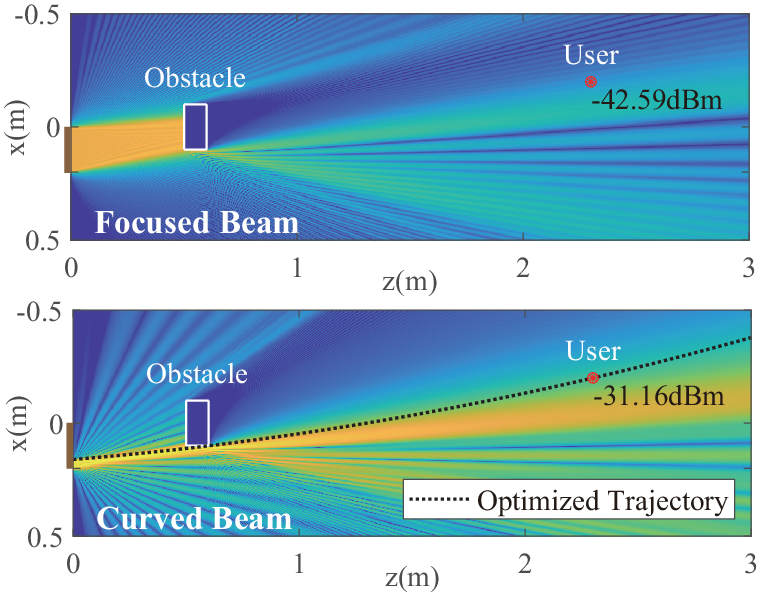}
				\caption{Power distribution of the optimized curved beam and the focused beam for $z_r = 2.3$ }
				\label{fig.cmp_curve_focused}
			\end{minipage}
			\begin{minipage}[b]{.3\linewidth}
			   \setlength{\belowcaptionskip}{-15mm}
					\centering
				\includegraphics[width=0.99\linewidth]{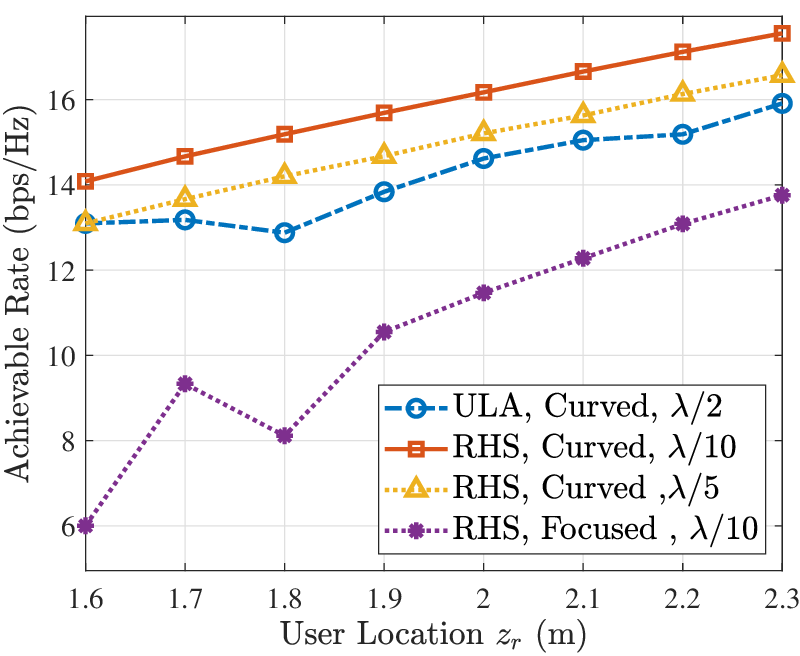}
				\caption{Achievable rate vs. the user's location $z_r$ for $x_r=-0.2$~m }
				\label{fig.simu_rate}
			\end{minipage}
			\begin{minipage}[b]{.34\linewidth}
				\setlength{\belowcaptionskip}{-15mm}
				\centering
				\includegraphics[width=0.99\linewidth]{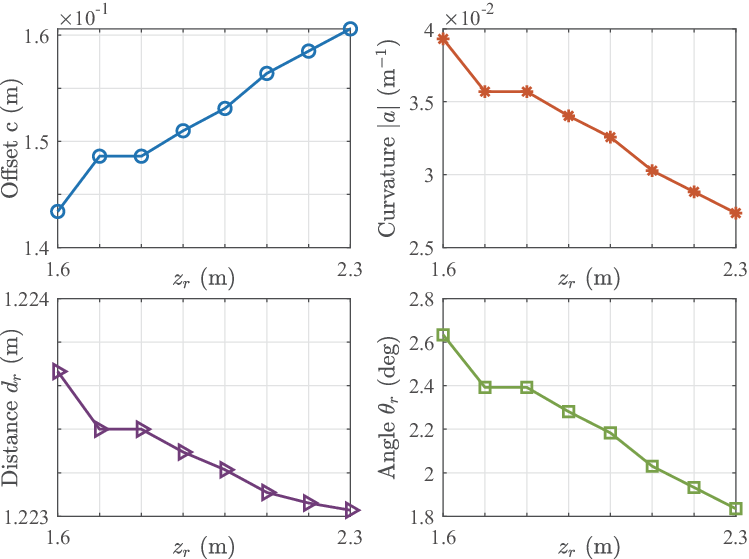}
				\caption{Optimized trajectory parameters vs. the user's location $z_r$}
				\label{fig.last}
			\end{minipage}
		\end{figure*} 
		
		In this section, we evaluate the performance of the proposed holographic Airy beamforming method and the geometry-based trajectory optimization algorithm. The carrier frequency is set to 100 GHz. The RHS aperture spans from $0$ m to $0.2$ m with an element spacing of $\lambda/10$. An obstacle is located within the region $[-0.1, 0.1]$~m $\times [0.5, 0.6]$~m. The user moves along the line $x_r = -0.2$~m, with $z_r$ varying from $1.6$~m to $2.3$~m. Throughout this movement, the LoS path between the user and the antenna remains completely blocked.
		
	
	Fig.~\ref{fig.cmp_curve_focused} illustrates the power distribution of the focused beam and the curved beam generated by Algorithm 1 when the user is located at $(-0.2,2.4)$~m. It can be observed that when the focused beam is used, the transmitted power, which is concentrated along the LoS path, is blocked by the obstacle, thus placing the user within a shadow zone. In contrast, when the curved beam is employed, the majority of the power bypasses the obstacle along the optimized trajectory, improving the received signal power at the user.


Fig.~\ref{fig.simu_rate} shows the achievable rate as the user moves from $z_r = 1.6$~m to $z_r = 2.3$~m. The curved beam can achieve a higher rate than the focused beam, as it enhances the user's received power, which is shown in~Fig.~\ref{fig.cmp_curve_focused}. Moreover, RHS utilizes an adjustable effective aperture to achieve flexible trajectory optimization, resulting in a higher rate than ULAs.  Besides, it is seen that the smaller element spacing of the RHS can achieve higher achievable rates due to higher-resolution wavefront engineering.


Fig.~\ref{fig.last} depicts the optimized trajectory parameters, offset $c$ and curvature $a$, along with the distance $d_r$ from the trajectory's maximum propagation point $(x_{max},z_{max})$ to the user and the angle $\theta_r$ between the user-$(x_{max},z_{max})$ line and the tangent, during the user's movement. The results indicate that  as the user moves from near to far, the trajectory offset $c$ varies within the RHS aperture, generating a trajectory with reduced curvature $a$. This leads to a gradual decrease in both the distance $d_r$ and the angle $\theta_r$, which in turn improves the user’s received power, as evidenced by the enhanced achievable rate presented in Fig.~\ref{fig.simu_rate}. This further corroborates the validity of the approximated received power given by (\ref{eq.approx}), which increases with the decline in both $d_r$ and $\theta_r$.

		\section{conclusion}
		In this paper, we have considered a downlink RHS-enabled communication system where the BS equipped with an RHS transmit to the user fully blocked by an obstacle. To  circumvent blockage, we have proposed  the holographic airy beamforming scheme to generate a curved beam following a parabolic trajectory. A geometry-based trajectory optimization algorithm has been developed to maximize the received power of the user.

		Three conclusions can be drawn from the simulation results. \emph{First}, the curved beam generated via holographic Airy beamforming enhances the received power for users blocked by obstacles compared to conventional focused beamforming. \emph{Second}, owing to its adjustable effective aperture, the RHS offers greater flexibility in trajectory design and achieves a transmission performance over 10 dB superior to that of the phase-controlled arrays.
		\emph{Third}, a decrease in element spacing for an RHS facilitates high-resolution wavefront engineering and precise beamforming, leading to higher user received power. 
		\appendices
		\section{}
	    We demonstrate that the radiation models given by (\ref{eq.rad_element}) and (\ref{eq.RHS_model}) are mutually convertible.  When model given  by (\ref{eq.rad_element}) is used to represent the radiation field, we  define $m_n^{'} = \sqrt{\prod_{k=1}^{n-1}(1 - \eta_k)\eta_n}$. By setting $m_n =m_n^{'} /\max{(m_n^{'} )}$, $s_n = 1$ and $\eta_{eq} = 1/\sum{m_n^2}$, we obtain the representation in model given by (\ref{eq.rad_element}).  When model given  by (\ref{eq.rad_element}) is used, we define $\eta_1 = m_1^2s_1\eta_{eq}$ and, for $n\ge2$,  $\eta_n = m_n^2 s_n\eta_{eq}/(1 - \eta_{eq}\sum_{k=1}^{n-1}m_k^2s_k)$. This enables the radiation field of the RHS to be written as the model given by (\ref{eq.rad_element}).
	    
	    \section{}
	    To solve (\ref{eq.function_need_to solve}), the key idea is to link the aperture coordinate $x$ to a point on the parabolic curve by using the tangent line geometry. For a point $(x_0,z_0)$ on the parabola, the intersection of its tangent line with the input plane at $z = 0$ determines the originating aperture point as $x = x_0 - z_0\frac{df(z_0)}{dz_0} = (az_0^2 + bz_0 + c) - z_0(2az_0 + b) = -az_0^2 + c$. Considering forward propagation where $z_0\ge0$, the aperture coordinate should satisfy $a(c-x)\ge 0$, which leads to the relation  $z_0 = \sqrt{\frac{c-x}{a}}$. Therefore,the slope at this corresponding curve point can be expressed as  a function of $x$, defined as $\psi(x)=\frac{df(z_0)}{dz_0} = 2a\sqrt{\frac{c-x}{a}} + b$. Substituting this expression for $\psi(x)$ into equation (\ref{eq.function_need_to solve}) transforms it into a differential equation solely in terms of $x$, and its solution for the phase profile, given in equation (\ref{eq.phase_profile}), is subsequently obtained through integration by parts.


		\bibliographystyle{IEEEtran}
		\bibliography{IEEEabrv.bib,reference.bib}

		\end{document}